\documentclass[a4paper, 10pt,twocolumn,twoside]{IEEEtran}

\usepackage{amssymb,amsmath,amsfonts}
\usepackage{color}
\usepackage{mathtools}
\usepackage{mathrsfs}
\usepackage{bm}
 \usepackage{graphicx}		
\usepackage{epstopdf}		
\usepackage[font={small}]{caption}
\usepackage{caption}
\usepackage{subcaption}
\usepackage{multicol}
\usepackage{enumerate}
\usepackage{units}
\usepackage[numbers]{natbib}
\usepackage{algorithm}
\usepackage{color}

\usepackage{algorithm}
\usepackage{algpseudocode}
\usepackage{cuted}

  \usepackage{amsthm}

\makeatletter
\def\BState{\State\hskip-\ALG@thistlm}
\makeatother

\usepackage{flushend}

\newtheorem{prep}{Proposition}

\newtheorem{cora}{Corollary}

\begin{document}

\title{On the SINR Distribution of SWIPT MU-MIMO with Antenna Selection}

 \author{Hadi Saki, M. Shikh Bahae \\
\IEEEauthorblockA{Institute of Telecommunications, King's College London, London, United Kingdom\\
Email: \{hadi.saki,m.sbahaei\}@kcl.ac.uk}}
\vspace{-3em}

\thispagestyle{empty}
\maketitle
\pagestyle{empty}
\begin{abstract}
In this paper, we provide a closed-form approximation of the Full-duplex Simultaneous wireless information and power transfer (SWIPT) multi-user MIMO (MU-MIMO) system signal-to-interference-and-noise ratio (SINR) distribution for the received signal at the sensor users with perfect and imperfect channel state information (CSI) and transmit antenna selection scheme at the transmitter. We also studied the SINR distribution of the uplink with Zero-Forcing beamforming (ZFBF) and antenna selection at the aggregator (AGG) receiver with transmit antenna selection.  The downlink SINR with perfect and imperfect CSI is modeled by a multivariate Beta type II distribution and the uplink with perfect and imperfect CSI is modeled with multivariate Wishart distributions. We compared the proposed distributions analytical results to the  Monte-Carlo simulations and we obtained a perfect match.


\end{abstract}


\section{Introduction}
\thispagestyle{empty}
Real-time and reliable data transmission is an essential requirement for next-generation autonomous platforms where power consumption, latency, and delay are important parameters . 
In terms of communication systems, research has addressed these requirements in several ways including cross layer design \cite{Shadmand2010,Nehra2010}, energy efficient communication technologies \cite{Olfat,Mahyari2015}, and improved reliable and fast communication systems. Recently, Full Duplex (FD) \cite{Towhidlou2016,Naslcheraghi2017} Simultaneous Wireless Information and Power Transfer (SWIPT) multi-user MIMO (MU-MIMO) has become a promising technology in wireless
communication systems due to its application in the short-range in-body and out-body wireless sensor networks. However, a particular challenge in implementing these methods is how to define and evaluate the outage performance and the ergodic capacity of the system. To this end deriving a simple closed-form SINR distribution is essential. Since the introduction MU-MIMO, many works have been pursued to derive a closed-form approximation of the outage capacity where the emphasis is on how to avoid the interference of the MU-MIMO network.

The closed-form expressions for pdf of the maximum eigenvalue of a product of multivariate independent complex Gaussian in single user MIMO is driven in single user MIMO Channels is derived in \cite{ShiJin2008}. Authors in \cite{Malla2017} proposed a closed-form model for the distributions of constrained SINR in MIMO under imperfect CSI at the transmitter. The work in \cite{Caire2010} studied the MU-MIMO fading broadcast channel and the achievable ergodic capacity is determined when CSI is obtained at the receivers and it is delivered to the transmitter.  
Furthermore, ZFBF is also has become a promising beamforming technique due to the performance in the imperfect channel knowledge availability compared to the matched filter (MF),  block diagonalization (BD) and singular value decomposition (SVD) beamforming techniques. SNR distribution of the uplink MIMO systems with ZFBM with perfect and imperfect CSI is calculated in \cite{Nosrat-Makouei2011}.  Generalized Selection Criterion performance evaluation of ZF-Precoded MU-MIMO also studied in \cite{Lee2013}. In this work, the SNR distribution is calculated to evaluate the ergodic sum-rate of MU-MIMO. Even though the MU-MIMO SINR distribution is derived in the previous work, they have neglected the antenna selection and energy harvesting performance and little attention have been paid to performance analysis of SWIPT MU-MIMO. Thus, we propose a closed-form approximation of the downlink SWIPT MU-MIMO and uplink  MU-MIMO with ZF-precoded with received antenna selection system SINR using random matrices.
\subsection{Contributions of our Work}

 In this work, we investigate the problem in an IBFD SWIPT MU-MIMO system. Our objective is to derive a closed-form distribution of the SINR in the present of the energy harvesting and antenna selection.
The primary contributions of this paper can be summarized as follows,
\begin{itemize}
\item We derived a new closed-form multivariate Beta distribution of type II formulation for the received SINR of the SWIPT MU-MIMO at the SUs under perfect and imperfect CSI.
\item  A closed-form low complexity multivariate Wishart distribution formulation of the SWIPT MU-MIMO SINR at the AGG with ZFBM under perfect and imperfect CSI is calculated.
\end{itemize}
\subsection{Organization and Notation}
The rest of this work is as follows. The system model is given in Section II.  SINR distribution of the downlink SWIPT MU-MIMO for perfect and imperfect CSI is presented in section III and in Section IV the SINR distribution of uplink MU-MIMO with antenna selection and ZF beamforming is presented. Finally, simulation results and the conclusion are discussed in Section V and Section VI.

\textit{Notation}: We use the following notations throughout this paper; bold upper case and lower case letters are used for matrices and vectors respectively, while small normal letters are kept for scalars. $\textbf{X}^*$, $\textbf{X}^{-1}$, $\textbf{X}^{\dagger}$, $Tr(\textbf{X})$, $\textbf{X}^{\mathcal{H}}$, and $\textbf{X}^{1/2}$, are used to denote the complex conjugate, transpose, inverse, pseudo-inverse, trace, Hermitiant and the square-root of $\textbf{X}$ respectively.  $\textbf{I}$ and $\textbf{O}$ denote an identity matrix and an all-zero matrix, respectively, with appropriate dimensions; $\textbf{X} \succeq \textbf{O}$ and $\textbf{X} \succ \textbf{O}$  mean that $\textbf{X}$  is positive semi-definite and positive definite, respectively;  $\mathbb{E}[\cdot]$ denotes the statistical expectation; The distribution of a CSCG random variable with zero mean and variance σ2 is denoted as $\mathcal{C}\mathcal{N}(0,\sigma^2) $ , and $\sim$ means ‘distributed as’; $\mathbb{C}^{x\times y}$   denotes the space of $x\times y$ matrices with complex entries; $\Arrowvert \textbf{x} \Arrowvert$ denotes the Euclidean norm of a vector $\textbf{x}$; the unit-norm vector of a vector $\textbf{x}$ is denoted as  $\vec{\textbf{x}}=\textbf{x}/{\Arrowvert \textbf{x} \Arrowvert}$; the quantity min(x,y) and max(x,y) represents the minimum and maximum between two real numbers.
\section{System Model}
We consider a bidirectional FD multi-user MIMO system as shown in fig. \ref{fig:FD}, where $K$ sensor user (SU) indexed with $k \in \mathcal{K} \triangleq \{1, 2, ..., K\}$ and each sensor is equipped with $N_U$ antennas communicate with an AGG equipped with $N_T$ and $N_R$ transmit and receive antennas respectively, where similarly indexed with $n_t \in \mathcal{N_T} \triangleq \{1, 2, ..., N_T\}$ and $n_r \in \mathcal{R} \triangleq \{1, 2, ..., N_R\}$. Without loss of generality, we Assume equal number of transmit and receive antennas at the AGG and the SUs' are equipped with small number of antennas compared to the AGG, i.e. $N_T=N_R \gg N_U$. The AGG is connected to a constant power supply. Uncorrelated antennas are assumed at the AGG. The SUs are energy limited devices and harvest their energy from transmitted signal by the AGG. SUs can split the received signal by using power splinter into two different energy harvesting (EH) and information detection (ID) elements. The power splitting (PS) ratio of the $k_{th}$ SU for the EH and ID elements are denoted by $\rho$ and $1-\rho$ respectively. SUs are also equipped with a limited capacity rechargeable battery that store the harvested energy.
\section{Downlink SWIPT MU-MIMO}
\subsection{Perfect channel knowledge at the AGG}
\begin{figure}
\centering
    \includegraphics[width=0.2\textwidth]{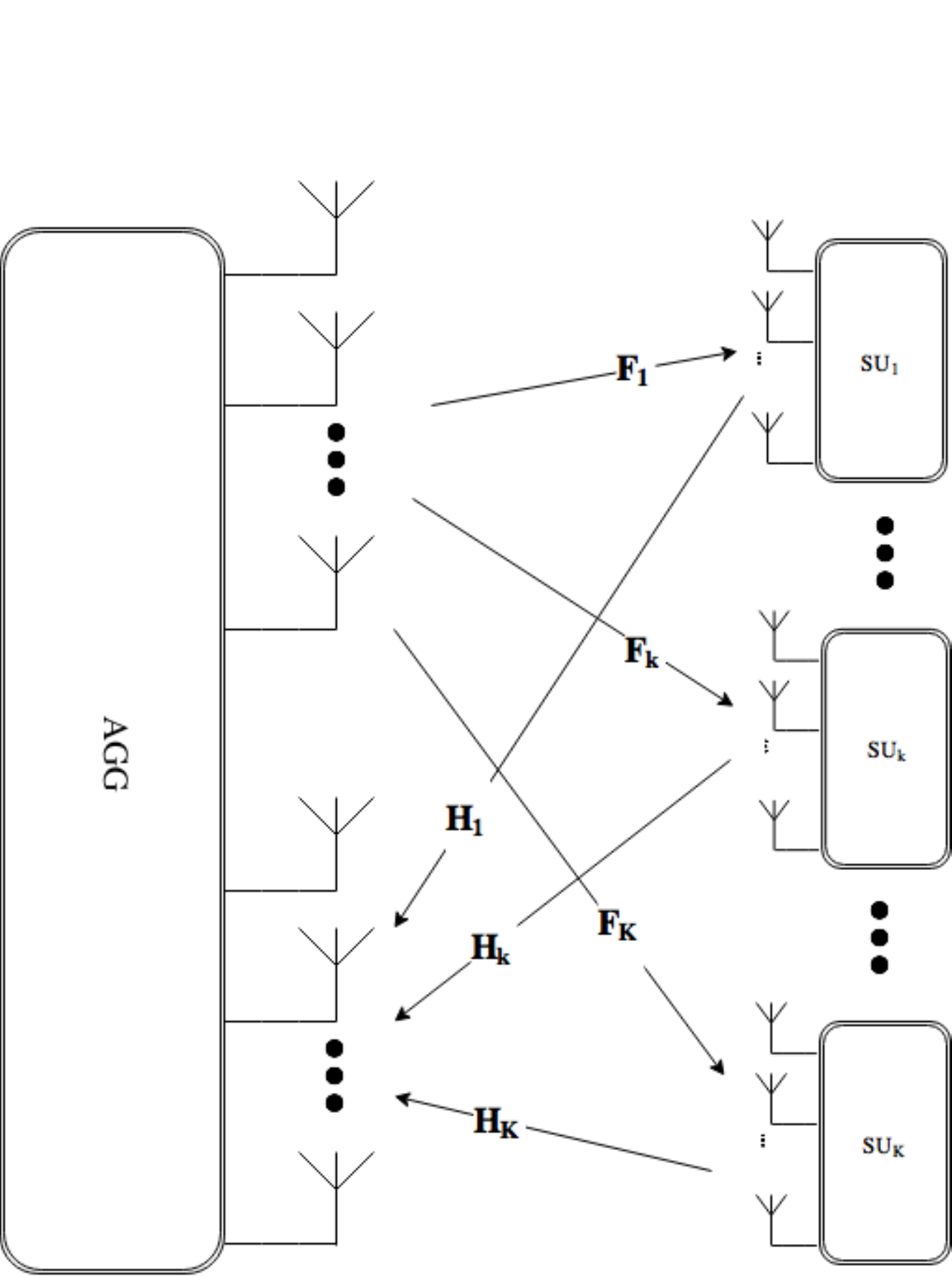}
  \caption{Full duplex mimo system model.}
  \label{fig:FD}
  \vspace{-1em}
\end{figure}
We consider an SWIPT system as shown in Fig. \ref{fig:FD}. The transmitted signal by $k_{th}$ SU is given by 
\begin{align}
\bf{x}_k^u=\bf{w}_k^u\bf{s}_k^u,
\label{eq:u_t}
\end{align}
where $\bf{x}_k^u$ is the transmitted signal, $\bf{w}_k^u \in \mathbb{C}^{N_U\times 1} $ is the beamforming vector and $\bf{s}_k^u \in \mathbb{C} $ and $\mathbb{E}[|{\bf{s}_k^u}|^2]=1$ is the transmitted information symbol. 
In this section we consider a perfect channel knowledge flat fading channel at the AGG. Therefore the received signal at the AGG can be modeled as, 
\begin{align}
\bf{y}_k^u\triangleq&~ {\sqrt{p^u}{\bf{H}_k\bf{w}_k^u\bf{s}_k^u}} +\underbrace { {\sqrt{p^u}\sum_{i\neq k, i\in K}{\bf{H}_i\bf{w}_i^u\bf{s}_i^u}}}_{{\sf\text { UL interference}}} \notag \\&+\underbrace {{\sqrt{p^d}\bf{G}^d\bf{w}^d\bf{s}^d}}_{{\sf\text { DL interference}}}+\bf{n}_k^u,
\label{eq:u_y}
\end{align}
where $p^u$ is the uplink maximum average transmit power, $\bf{H}_k \in \mathbb{C}^{N_r\times N_u}$ denotes the channel matrix between the $k_{th}$ SU and AGG. Downlink transmit power is denoted as $\sqrt{p^d}$,  and $\bf{w}^d$ is the downlink precoding matrix. $\bf{G}^d \in \mathbb{C}^{N_r\times N_t}$ is the self-interference (IS) noise at the AGG, $N_r$,  $N_t$ and $N_u$ are the number of receive and transmit antenna at the AGG and the number of FD antennas at the user respectively. $\bf{s}^d$ is the downlink symbol vector, and $\bf{n}^u \in \mathbb{C}^{N_r\times 1} $ is the additive Whit Gaussian noise (AWGN) at the AGG. 
To be able to reduce the computational complexity and also to  reduce the energy consumption at the AGG with large number of receive antenna arrays, we use the antenna selection in the AGG.  We select $N_s$ antennas out of total $N_r$ receive antennas in a way that $N_s \ll N_r$. This is  mathematically 
applied by the use of the diagonal antenna selection matrix $\bf V_k^u\in \mathbb{C}^{N_r\times N_r}$ , where $[\bf V]_{i,i} = 1$ if the antenna $i$ is selected and $[\bf V]_{i,i} = 0$ if the $i_{th}$ antenna is not selected,  $\sum_{i=1}^{N_r}[\bf {V}]_{i,i}= {N_R} $. Therefor the related received signal will be  denoted as
\begin{align}
\bf{y}_{k,s}^u\triangleq&~ {\sqrt{p^u}\bf V_k^u{\bf{H}_k\bf{w}_k^u\bf{s}_k^u}} +\underbrace { {\sqrt{p^u}\sum_{i\neq k, i\in K}{\bf V_i^u\bf{H}_i\bf{w}_i^u\bf{s}_i^u}}}_{{\sf\text { UL interference}}} \notag \\&+\underbrace {{\sqrt{p^d}\bf V^u \bf{G}^d\bf{w}^d\bf{s}^d}}_{{\sf\text { DL interference}}}+\bf V_k^u \bf{n}_k^u
\label{eq:u_y1}.
\end{align}
We denote $\bf V_k^u \bf n_k^u=\bf n_{k,s}^u$, therefore the received SINR in the AGG can be expressed as eq.\ref{eq:u_sinr_1}. 



\newcounter{storeeqcounter3}
\newcounter{tempeqcounter3}

\addtocounter{equation}{0}%
\setcounter{storeeqcounter3}{\value{equation}}%

 \begin{figure*}[!t]

 \normalsize

\setcounter{tempeqcounter3}{\value{equation}} 
\setcounter{equation}{\value{storeeqcounter3}} 

\begin{align}
{\gamma}_{k,s}^{u}=\frac{{p^u}\bf {\bf{H}_k^{\mathcal{H}}\bf{V_k^u}^{\mathcal{H}}{\bf{w}_k^u}^{\mathcal{H}}\bf{w}_k^u} \bf V_k^u\bf{H}_k}{{p^u}\sum_{i\neq k, i \in K}{ \bf{H}_i^{\mathcal{H}}\bf{V_i^u}^{\mathcal{H}}{\bf{w}_i^u}^{\mathcal{H}}\bf{w}_i^u}\bf V_i^u \bf{H}_i+{{p^d}{\bf{G}^d}^{\mathcal{H}} \bf{V^u}^{\mathcal{H}}{\bf{w}^d}^{\mathcal{H}}\bf{w}^d}\bf V^u\bf{G}^d+{\sigma_k^u}^2}.
\label{eq:u_sinr_1}
\end{align}




\begin{align}
{\gamma}_{k,s}^{u,zf}=\frac{(1-\alpha_k^2)P_k\bf{V_k^u}^\mathcal{H} \bf {w_k^u}^\mathcal{H} \bf{w}_k^u \bf V_k^u/ d_k }{  \bf {Z_k}^\mathcal{H}(tr(\alpha^2  \bf {V^u}^\mathcal{H} {\bf w^u}^\mathcal{H}  P \bf w^u \bf V^u) (\check{\bf{H}}_k^\mathcal{H} \check{\bf H}_k)^{-1}+{\sigma_n^u}^2(\check{\bf{H}}_k^\mathcal{H} \check{\bf H}_k)^{-1}   )\bf{Z}_k},
\label{eq:u_sinr5}
\end{align}

 \hrulefill

\vspace{-1em}
 \end{figure*}


The eq. \ref{eq:u_y1} can be rearranged as, 
\begin{align}
\bf y_{k,s}^u= \tilde{\bf{H}}
\begin{bmatrix}
\bf V_k^u \bf w_k^u \bf{s}_k^u \\
\sum_{i\neq k, i\in K}{\bf{C}_i}\bf{s}_i^u + \bf \tilde{\bf G}^d \bf s^d
\end{bmatrix} 
+\bf{n}_{k,s}^u,
\label{eq:u_y2}
\end{align}
where $\tilde{\bf{H}}_k$$\triangleq \left[ \bf{H}_k  \quad \bf{A}_i\right] $, $\bf{A}_i \bf{C}_i=\bf V_i^u \bf{H}_i\bf{w}_i^u, \:  \{\forall\: i\in K \: \& \: i\neq k\}$, and $\bf A_i \tilde{\bf G}^d=\bf V^u \bf G^d \bf w^d$ are the uplink and downlink interference coefficient respectively. 

To derive the received signal distribution at the AGG, we first assume that the transmit matrix is fixed, and then the zero-forcing (ZF) processing technique is used for the data reception at AGG. The ZF beamforming equalizer at the AGG receiver can be implemented by precode filter \cite{Yuk-FanHo2008}. The ZF is a standard method widely implemented in the massive MIMO processing systems \cite{Bergel2016,Hwang2016}.
Furthermore, since $\bf{H}_k$ is not correlated with $\bf A_i$ and $\bf w_k^d$, therefore 
the $N_r \times N_r$ matrix of  $[\bf H_k \quad A_i]$  is full rank. Consequently, the ZF beamforming equalizer at the AGG receiver can be expressed as,
\begin{align}
\bf U_{k}^{zf}\!\!=\!\!\left[\bf I_{d_k} \quad \bf 0 \right] \tilde{\bf H}_k^{-1},
\label{ZF}
\end{align}
where we define $\bf Z_k = \left[\bf I_{d_k} \quad \bf 0 \right]$. 
\begin{prep}
For multi user uplink FD MIMO system with perfect CSI and adopting ZF beamforming $\bf U_k^{zf}$, the processed SINR follows an $N_r\times N_r$ complex central Wishart distribution with identity covariance matrix of the  $rank(\bf{U}_k^u \bf{H_k})=d_k$, and $2N_u$ degree of freedom.
\end{prep}
\begin{proof}
By applying the ZF beamforming to eq.\ref{eq:u_y} the $k_{th}$ user received signal at \ref{eq:u_y2} changes to,  
\begin{align}
\bf{y}_k^{u,zf} \!= \!\!\bf Z_k \tilde{H}_k^{-1}\tilde{\bf{H}}
\begin{bmatrix}
\bf V_k^u \bf w_k^u \bf{s}_k^u \\
\sum_{i\neq k, i\in K}{\bf{C}_i}\bf{s}_i^u + {\bf \tilde{\bf G}}^d \bf s^d
\end{bmatrix} 
 \!\!+ \!\!\bf Z_k \tilde{H}_k^{-1}\bf{n}_{k,S}^u,
\label{eq:u_y2}
\end{align}
and the $k_{th}$ user received SINR at eq.\ref{eq:u_sinr_1} will be,
\begin{align}
{\bf{\gamma}}_{k,s}^{u,zf}=\frac{p\bf {V_k^u}^\mathcal{H}\bf {w_k^u}^\mathcal{H} \bf{w}_k^u \bf V_k^u}{\bf Z_k(\tilde{\bf{H}}_k^\mathcal{H} \tilde{\bf H}_k)^{-1} \bf{Z}_k^\mathcal{H} d_k {\sigma_n^u}^2}.
\label{eq:u_sinr2}
\end{align}
we define  $\boldsymbol\eta_k^u= \frac{\sqrt{p}\bf V_k^u\bf {w_k^u}}{\sqrt{d_k{\sigma_n^u}^2}} $. From \cite{Watkins2010} we have,
\begin{align}
\bf Z_k(\tilde{\bf{H}}_k^\mathcal{H} \tilde{\bf H}_k)^{-1} \bf{Z}_k^\mathcal{H}=\bf Z_k\left(
\begin{bmatrix}
\bf{H}_k^{\mathcal{H}} \bf {H}_k &   \bf{H}_k ^{\mathcal{H}} \bf{A_i} \\
\bf{A_i^{\mathcal{H}}} \bf{H}_k & \bf{A}_i^{\mathcal{H}} \bf {A}_i
\end{bmatrix}
\right)^{-1}\!\! \bf{Z}_k^\mathcal{H}.
\end{align}
Since multiplication of $(\tilde{\bf{H}}_k^\mathcal{H} \tilde{\bf H}_k)^{-1}$ by $\bf Z_k$ and $\bf Z_k^{\mathcal{H}}$ only keeps the $d_k$ row and columns respectively, the  $\bf Z_k(\tilde{\bf{H}}_k^\mathcal{H} \tilde{\bf H}_k)^{-1} \bf{Z}_k^\mathcal{H}$ can be simplified to ($\bf{H}_k^{\mathcal{H}} \bf {H}_k- \bf{H}_k ^{\mathcal{H}} \bf{A_i} \bf{A_i^{\mathcal{H}}} \bf{H}_k)^{-1} =(\bf{H}_k^{\mathcal{H}}(\bf I_{N_r}-\bf A_i \bf A_i^{\mathcal{H}})\bf H_k)^{-1}$, where $(\bf I_{N_r}-\bf A_i \bf A_i^{\mathcal{H}})$ is the projection matrix of the rank $d_k$. Since , $\bf{H}_k \in \mathbb{C}^{N_r\times N_u}$ is a complex Gaussian random vector with zero mean and unit variance, $\bf{H}_k^{\mathcal{H}} \bf {H}_k $ follows a central Wishart distribution with $\bf I_{N_r}$ covariance matrix and $2N_r$ degree of freedom.
Therefore, considering   $(\bf I_{N_r}-\bf A_i \bf A_i^{\mathcal{H}})$ as a projection matrix of $d_k$, then  $(\bf{H}_k^{\mathcal{H}}(\bf I_{N_r}-\bf A_i \bf A_i^{\mathcal{H}})\bf H_k)$ distributed as Wishart distribution with $\bf I_{d_k}$ covariance matrix and $2N_r$ degrees of freedom. Therefore, according to \cite{citeulike:624837} (theorem 3.3.11), SINR has a Wishart distribution with ${\bf \boldsymbol \eta_k^u}^\mathcal{H} \bf I_{d_k} {\boldsymbol \eta_k^u}$ covariance matrix and $2N_r$ degrees of freedoms,\vspace{-0.5em}
\begin{align}
f(\gamma_{k,s}^{uzf})\!\!=\!\!\frac{{\gamma_{k,s}^{uzf}}^{(2N_r-N_u-1)/2}\!exp\!\left[\!-\frac{1}{2}tr(\bf\! {({\boldsymbol \eta_k^u}^\mathcal{H} I_{d_k} {\boldsymbol \eta_k^u})}^{-1}\gamma_{k,s}^u)\!\right]\!}{(2)^{N_rN_u}\Gamma_{N_{uzf}}(N_r)\det({\boldsymbol \eta_k^u}^\mathcal{H} \bf I_{d_k} {\boldsymbol \eta_k^u})^{N_r}},\vspace{-0.5em}
\end{align}
where $\Gamma(.)$ is a multivariate gamma function and $(N_r-N_u-1)\gg0$ with regard to Lebesque measure of positive definite cone.\vspace{-1em}
\end{proof}
\subsection{Imperfect channel knowledge at the AGG}
 In the realistic wireless communication system, due to the feedback delays and/or estimation error, actual
channel  is different from estimated channel and can be modeled as, 
\begin{align}
\bf H_k=\sqrt{1-\alpha^2}\bf \hat{H}_k+\alpha \bf \Delta_k,
\end{align}
where $\hat{H}_k\sim \mathcal{C}\mathcal{N}(0,\bf I)$ is the imperfect estimated channel with zero mean and unit variance at the AGG and $\Delta_k \sim \mathcal{C}\mathcal{N}(0, \bf I)$ is the estimated channel Gaussian noise at AGG. The received signal at the AGG after applying the antenna selection is denoted as,
\begin{align}
\bf{y}_k^u\triangleq&~ {\sqrt{p^u}{(\sqrt{1-\alpha^2}\bf \hat{H}_k+\alpha \bf \Delta_k)\bf V_k^u\bf{w}_k^u\bf{s}_k^u}} \notag \\&+ { {\sqrt{p^u}\sum_{i\neq k, i\in K}{(\sqrt{1-\alpha^2}\bf \hat{H}_i+\alpha \bf \Delta_i) \bf V_i^u\bf{w}_i^u\bf{s}_i^u}}} \notag \\&+ {{\sqrt{p^d} \bf{G}^d \bf V^u\bf{w}^d\bf{s}^d}}+\bf V_k^u\bf{n}_k^u \notag \\ 
&=\sqrt{1-\alpha_k^2} \check{H}_k
\begin{bmatrix}
\bf V_k^u \bf w_k^u \bf s_k^u \\ \\
\sum_{i\neq k, i\in K}{\bf{\hat{C}}_i}\bf{s}_i^u  + \bf \check{\bf G}^d \bf s^d
\end{bmatrix}
+ \notag \\ & \alpha_k \Delta \bf V^u\bf w^u \bf s^u+\bf n_{k,s}^u,
\label{eq:u_y_im}
\end{align}
where $\bf \check{H}_k\triangleq \left[ \bf{\hat{H}}_k  \quad \bf\check{A}_i\right]$, $\bf\check{A}_i \bf\hat{C}_i=\bf\hat{H}_i\bf V_i^u \bf{w}_i^u, \:  \{\forall\: i\in K \: \& \: i\neq k\}$, and $\bf\check A_i \check{\bf G}^d=\bf G^d \bf V^u \bf w^d / \sqrt{1-\alpha_k^2}$ are the uplink and downlink interference coefficient respectively, $\Delta=[\Delta_1,\: \Delta_2,\: .\:.\:.\: \Delta_K]$, $\bf w^u=[\bf w_1^u,\: \bf w_2^u,\: .\:.\:.\: \bf w_K^u]$ and $\bf s^u=[\bf s_1^u,\: \bf s_2^u,\: .\:.\:.\: \bf s_K^u]$. 
Similar to the eq.\ref{ZF}, the ZF beamforming equalizer at the AGG receiver can be expressed as,
\begin{align}
\bf \check U_k^{zf}=\left[\bf I_{d_k} \quad \bf 0 \right] \check{\bf H}_k^{-1}=\bf \check{Z}_k\check{\bf H}_k^{-1}.
\label{ZF2}
\end{align}
\begin{prep}
For multi user uplink FD MIMO system with imperfect CSI at AGG and adopting ZF beamforming $\bf \check {U}_k^{zf}$, the processed SINR distribution is as follows:
\begin{align}
f(\gamma_{k,s}^{uzf})\!\!=\!\!\frac{{\gamma_{k,s}^{uzf}}^{(2N_r\!-\!N_u\!-\!1)/2}\!\!exp\!\left[\!-\frac{1}{2}tr(\bf ({\boldsymbol\eta_k}^\mathcal{H}{I_{dk}}{\boldsymbol\eta_k})^{-1}\gamma_{k,s}^u)\right]}{(2)^{N_rN_u}\Gamma_{N_{uzf}}(N_r)\det(\bf {{\boldsymbol\eta_k}^\mathcal{H}I_{d_k}{\boldsymbol\eta_k}})^{N_r}},
\label{sinr_zf_i}
\end{align}
where $\boldsymbol\eta_k=\frac{\sqrt{(1-\alpha_k^2)P_k}(  \bf V_k^u \bf w_k^u)}{\sqrt{d_k(\alpha_k^2 P \mathcal{J}+\sigma_k^{u})}}$ and $\mathcal{J}=tr({\bf V^u}^ \mathcal{H} {\bf w^u}^\mathcal{H} \bf w^u \bf V^u/d)$. and $d=\sum_i=1^K d_i$
\end{prep}
\begin{proof}
The post-processing received signal at the AGG after applying the \ref{ZF2} ZF beamforming is given by, 
\begin{align}
\bf{y}_{k,s}^{uzf}\triangleq&~\sqrt{1-\alpha_k^2} \bf Z_k {\bf{\check H_k}}^{-1}\check{H}_k
\begin{bmatrix}
\bf V_k^u \bf w_k^u \bf s_k^u \\ \\
\sum_{i\neq k, i\in K}{\bf{\hat{C}}_i}\bf{s}_i^u  + \bf \check{\bf G}^d \bf s^d
\end{bmatrix}
+ \notag \\ & \alpha \bf Z_k \bf {\check{H}}_k^{-1}\Delta \bf V^u\bf w^u \bf s^u+ \bf Z_k \bf {\check{H}}_k^{-1} n_{k,s}^u.
\label{eq:u_y_im}
\end{align}
Therefore the post-processing received SINR at the AGG for the $k_{th}$ user is given in eq.\ref{eq:u_sinr4}.












From \cite{ArthurFrederickSeber2008} (21.6), we have 
\begin{align}
\mathbb{E}(\alpha^2   ({\bf {\check{H}}_k^{-1}})^\mathcal{H}  \Delta^\mathcal{H} {\bf V^u}^\mathcal{H}  {\bf w^u}^\mathcal{H}  P \bf w^u \bf V^u \Delta {\bf {\check{H}}_k^{-1}}) = \notag \\ (tr(\alpha^2 \Delta^\mathcal{H} {\bf V^u}^\mathcal{H}  {\bf w^u}^\mathcal{H}  P \bf w^u \bf V^u \Delta)\Theta)  ({\bf {\check{H}}_k^{-1}})^\mathcal{H} {\bf {\check{H}}_k^{-1}} \notag \\ +\mathbb{E}(\Delta^\mathcal{H})(\alpha \Delta^\mathcal{H} {\bf V^u}^\mathcal{H}  {\bf w^u}^\mathcal{H}  P \bf w^u \bf V^u \Delta)\mathbb{E}(\Delta) ({\bf {\check{H}}_k^{-1}})^\mathcal{H} {\bf {\check{H}}_k^{-1}},
\end{align}
where $\Theta$ is the $\Delta$ variance matrix and $P=\mathbb{E}(\bf (s^u) ^\mathcal{H}\bf s^u)$. Since $\Delta$ is i.i.d with zero-mean and unit-variance elements, eq.\ref{eq:u_sinr4} can be simplified as eq.\ref{eq:u_sinr5}. It can be more simplified to,
\begin{align}
{\gamma}_{k,s}^{uzf}&=\frac{(1-\alpha_k^2)P_k \bf {V_k^u}^\mathcal{H}\bf {w_k^u}^\mathcal{H}  \bf{w}_k^u \bf {V_k^u}/ d_k }{  \bf {Z_k}^\mathcal{H}((\alpha^2 P \mathcal{J} +{\sigma_n^u}^2)(\check{\bf{H}}_k^\mathcal{H} \check{\bf H}_k)^{-1}   )\bf{Z}_k} \notag \\&= \frac{(1-\alpha_k^2)P_k \bf {V_k^u}^\mathcal{H} \bf {w_k^u}^\mathcal{H} \bf{w}_k^u \bf V_k^u / d_k }{ (\alpha^2 P \mathcal{J} +{\sigma_n^u}^2)(\hat{\bf{H}}_k^\mathcal{H} (\bf I_{N_r}-\bf \check{A}^\mathcal{H}\bf \check{A})\hat{\bf H}_k)^{-1}   }.
\label{eq:u_sinr5}
\end{align}
Therefore, similar to the perfect CSI case   $(\bf I_{N_r}-\bf \check{A}_i \bf \check{A}_i^{\mathcal{H}})$ is a projection matrix of rank $d_k$, then  $(\bf \check{H}_k^{\mathcal{H}}(\bf I_{N_r}-\bf \check{A}_i \bf \check{A}_i^{\mathcal{H}})\bf \check{H}_k)$ distributed as Wishart distribution with $\bf I_{d_k}$ covariance matrix and $2N_r$ degrees of freedom and also from \cite{citeulike:624837} (theorem 3.3.11) the received SINR distribution have also have a Wishart distribution with $\boldsymbol\eta_k^\mathcal{H}\bf I_{d_k} {\boldsymbol\eta_k}$ covariance and $2N_r$ degrees of freedom given in eq. \ref{sinr_zf_i}.\vspace{-1em}
\end{proof}
\section{Downlink SWIPT MU-MIMO}
\subsection{Downlink with perfect CSI }
Similar to the uplink scenario, here also we first consider a fixed beamforming at the transmitter and to derive the SINR at the AGG, the transmitted signal is expressed as,\vspace{-1.5em}
\begin{align}
\textbf{x}^d=\sum_{k=1}^K{\textbf{w}_k^d\textbf{s}_k^d},
\label{eq:d_t}
\end{align}
where $\bf{x}^d$ is the transmitted signal at the AGG, $\bf{w}_k^d \in \mathbb{C}^{N_t\times N_u} $ is the beamforming vector at the transmitter and $\textbf{s}_k^d \in \mathbb{C} $ is the transmitted information. We assume that $\mathbb{E}[|{\textbf{s}_k^d}|^2]=1$. The received signal at the $k_{th}$ SU can be expressed as, 
\begin{align} 
\bf{y}_{k}^{d}\triangleq&~\underbrace {\sqrt{p^d}{\bf{F}_k{\bf V_k^d}\bf{w}_k^d\bf{s}_k^d}} +\,\underbrace {\displaystyle \sum _{i\neq k, i \in K} \sqrt{p^d}{\bf{F}_i{\bf V_i^d}\bf{w}_i^d\bf{s}_i^d}}_{{\sf\text { DL interference}}} \notag \\[2pt]&+\,\underbrace {\displaystyle \sqrt{p^u}{\bf{G}^u\bf{w}^u\bf{s}^u}}_{{\sf\text { UL intracell interference}}} + \bf{n}_k^d. 
\end{align}
$\bf{F}_k \in \mathbb{C}^{N_u\times N_t}$ denotes the channel matrix between the AGG and the $k_{th}$ SU. $\bf{w}^u$ is the SUs' uplink beamforming matrix, and $\bf{G}^u \in \mathbb{C}^{K\times N_u}$ is the self-interference (IS) noise at the SUs. $\bf{s}^u$ is the uplink symbol vector, and $\bf{n}_k^d \in \mathbb{C}^{N_u\times 1} $ is the additive Whit Gaussian noise (AWGN) at the $k_{th}$ SU.
We assume that each SU antenna is equipped with a PS device that coordinates energy harvesting and the information decoding from the received signal. Therefore the received ID and EH signal elements are respectively denoted as,
\begin{align}
\textbf{y}_k^{di}=&\sqrt{\rho_k} [ {\sqrt{p^d}{\bf{F}_k{\bf V_k^d}\bf{w}_k^d\bf{s}_k^d}} &+\, {\displaystyle \sum _{i\neq k, i \in K} \sqrt{p^d}{\bf{F}_i{\bf V_i^d}\bf{w}_i^d\bf{s}_i^d}} \notag \\&+\, {\displaystyle \sqrt{p^u}{\bf{G}^u\bf{w}^u\bf{s}^u}}]+\bf{n}_k^s,
\label{eq:d_i}\\
\textbf{y}_k^{de}=&\sqrt{1-\rho_k}[ {\sqrt{p^d}{\bf{F}_k{\bf V_k^d}\bf{w}_k^d\bf{s}_k^d}} &+\, {\displaystyle \sum _{i\neq k, i \in K} \sqrt{p^d}{\bf{F}_i{\bf V_i^d}\bf{w}_i^d\bf{s}_i^d}} \notag \\&+\, {\displaystyle \sqrt{p^u}{\bf{G}^u\bf{w}^u\bf{s}^u}}+\textbf{n}_k^d],
\label{eq:d_e}
\end{align}
where $\textbf{n}_k^s$ is the power split processing noise at the $k_{th}$ SU. 
By assuming that full CSI is available at the AGG, the received signal to interference ratio (SINR) at $k_{th}$ SU is given and and the harvested power by $k_{th}$ SU are given in equations \eqref{eq:d_sinr} and \eqref{eq:d_energy} on top of Page~\pageref{eq:d_sinr}, respectively. Here $\eta_k \in (0,1]$ is the energy conversion efficiency of the is the user $k$ energy harvesting.
The received SINR at the $k_{th}$ user therefore can be noted as eq.\ref{eq:d_sinr}.
%
\newcounter{storeeqcounter}
\newcounter{tempeqcounter}
\addtocounter{equation}{0}%
\setcounter{storeeqcounter}{\value{equation}}%
%
\begin{figure*}[!t]
 \normalsize

\setcounter{tempeqcounter}{\value{equation}} 
\setcounter{equation}{\value{storeeqcounter}} 



\begin{align}
{\gamma}_k^{d}=\frac{{\bf{F}_k^{\mathcal{H}}\bf {V_k^d}^{\mathcal{H}}{\bf{w}_k^d}^{\mathcal{H}}\bf{w}_k^d}\bf {V_k^d}\bf{F}_k}{\left[\sum_{i=1, i\neq k}^{K}{\bf{F}_k^{\mathcal{H}}\bf {V_i^d}^{\mathcal{H}}{\bf{w}_i^d}^{\mathcal{H}}\bf{w}_i^d}\bf {V_i^d}\bf{F}_k+{{\frac{p^u}{p^d}}{\bf{G}^u}^{\mathcal{H}}{\bf{w}^u}^\mathcal{H}\bf{w}^u}\bf{G}^u+\frac{{\sigma_k^d}^2}{{p^d}}\right]+\frac{{\sigma_k^s}^2}{\rho_k{p^d}}},
\label{eq:d_sinr}
\end{align}
\begin{align}
{\gamma}_k^{d}=\dfrac{{\bf \hat{F}_k^{\mathcal{H}}\bf {V_k^d}^{\mathcal{H}}{\bf{w}_k^d}^{\mathcal{H}}\bf{w}_k^d}\bf {V_k^d}\bf\hat{F}_k}{[{{\sum_{i=1, i\neq k}^{K}{\bf{\hat{F}}_k^{\mathcal{H}}\bf {V_i^d}^{\mathcal{H}}{\bf{w}_i^d}^{\mathcal{H}}\bf{w}_i^d}\bf {V_i^d}\bf{\hat{F}}_k+\frac{\alpha^2}{(1-\alpha^2)}\sum_{i=1}^{K}{\bf{{\Delta}_i^d}^{\mathcal{H}}\bf {V_i^d}^{\mathcal{H}}{\bf{w}_i^d}^{\mathcal{H}}\bf{w}_i^d}\bf {V_i^d}\bf{{\Delta}_i^d}
+{{\frac{p^u}{(1-\alpha^2)p^d}}{\bf{G}^u}^{\mathcal{H}}{\bf{w}^u}^\mathcal{H}\bf{w}^u}\bf{G}^u}}}\notag\\ +\frac{{\sigma_k^d}^2}{(1-\alpha^2){p^d}}+\frac{{\sigma_k^s}^2}{\rho_k(1-\alpha^2){p^d}}]
\label{eq:d_sinr4}
\end{align}





 \hrulefill

 \end{figure*}

\begin{cora}
Assuming $\bf X$ is a multivariate Gaussian distribution matrix as $\mathcal{N}_d(\boldsymbol{\mu}_i,\boldsymbol \sigma)$ then $\bf \bf X^\mathcal{H} \bf Q
\bf X$ has a Wishart distribution denoted as $\bf \bf X^\mathcal{H} \bf Q
\bf X\approx W_d (r,\boldsymbol{\sigma},\boldsymbol{\delta})$, where $\bf Q$ is a symmetric $d\times d$ matrix and $i \in r$ is the number of the independent collumn s of $\bf X^{\mathcal{H}}$, and 
\begin{align}
\boldsymbol \delta = \boldsymbol\sigma^{-1/2}\bf M^{\mathcal{H}} \bf Q \bf M \boldsymbol\sigma^{-1/2}, \\ 
\bf M =(\boldsymbol{\mu_1},\boldsymbol{\mu_2},...,\boldsymbol{\mu_r}).
\end{align}
\end{cora}
\begin{prep}
For multi user downlink FD MIMO system with perfect CSI and the processed SINR follows a matrix variate Beta type $II$ distribution with parameters $(N1,N2)$ and defined as,
\begin{align}
&\boldsymbol\gamma_k^d\approx B_{N_u}^{II}(N1,N2)\sim \notag \\ &\frac{\det(\boldsymbol \gamma_k^d)^{(\frac{2N1-N_u-1}{2})}\det{(\bf I_{qk}+\boldsymbol{\gamma}_k^d)^{{-N1-N2}}}}{\beta( N1,N2)},
\label{gamma_d_p}
\end{align}
where\vspace{-1em}
\begin{align}
\beta(N1,N2)= \frac{\Gamma_{N_u}(N1)\Gamma_{N_u}(N2)}{\Gamma_{N_u}({N1+N2})},
\end{align}
and $\Gamma_{N_u}(x)$ is the multivariate Gamma function given as,
\begin{align}
\Gamma_{N_u}(x)=\pi^{{N_u}({N_u}-1)/4}\prod_{i=1}^{N_u}\Gamma(x-(i-1)/2).
\end{align}
\end{prep}
\begin{proof}
For simplicity we denote the received SINR of MU-MIMO with energy harvesting for $k_{th}$ user at  eq.\ref{eq:d_sinr} as $\boldsymbol\gamma_k^d=\frac{\boldsymbol\Phi}{\boldsymbol\Psi}$ and ${\boldsymbol\sigma_k^{d_0}}=\frac{{\boldsymbol\sigma_k^d}^2}{{p^d}}+\frac{{\boldsymbol\sigma_k^s}^2}{\rho_k{p^d}}$. Since , $\bf{F}_k \in \mathbb{C}^{N_u\times N_t}$ is a complex Gaussian random vector with zero mean and unit variance, $\bf{F}_k^{\mathcal{H}} \bf {F}_k $ follows a central Wishart distribution with $\bf I_{q_k}$ covariance matrix and $2N_t$ degree of freedom as, 
Therefore, considering $\bf{w}_k^d \in \mathbb{C}^{N_u\times N_t} $,   $ \bf {V_k^d}^{\mathcal{H}}{w_k^d}^{\mathcal{H}} \bf w_k^d \bf V_k^d$ as a projection matrix of $ q_k$, then  $\bf{F}_k^{\mathcal{H}}(\bf {V_k^d}^{\mathcal{H}} \bf {w_k^d}^{\mathcal{H}} \bf w_k^d \bf V_k^d)\bf F_k$ distributed as Wishart distribution with $\bf I_{q_k}$ covariance matrix and $2N_t$ degrees of freedom. Therefore, $\Phi$ has a central multivariate matrix Wishart distribution with $\bf I_{q_k}$ covariance matrix and $2N_t$ degrees of freedom as follows,
 \begin{align}
&\boldsymbol\Phi=\bf W^d\sim \bf \mathcal W_{N_u}(2N_t,\bf I_{qk})=\notag\\&\frac{{\gamma_{k}^{d}}^{(2N_t-N_u-1)/2}exp\left[-\frac{1}{2}tr(\bf ( I_{q_k})^{-1}\gamma_{k}^d)\right]}{2^{N_u N_t}\Gamma_{N_{u}}(N_t)\det( \bf I_{q_k})^{N_t}},
\end{align}
where $\bf \mathcal W_x(y,z)$ is a central multivariate Wishart distribution. Furthermore, the the first term of $\Psi$ in eq.\ref{eq:d_sinr}, is a sum of central Wishart distribution with similar identity covariance matrices can be simplified as follow,
\begin{align}
&\sum_{i=1, i\neq k}^{K}{\bf{F}_k^{\mathcal{H}}\bf {V_k^d}^{\mathcal{H}}{\bf{w}_i^d}^{\mathcal{H}}\bf{w}_i^d\bf{V}_i^d}\bf{F}_k= \notag \\  &\sum _{i=1,i\neq k}^K \bf W^l_i =\bf W^s\sim \bf \mathcal W_{N_u}(2N_t(K-1),\bf I_{qk}),
\end{align}
where $\bf W_i^l \sim \bf \mathcal W_{N_t}(2N_t,I_{qk})$ has a multivariate Wishart distribution with $\bf I_{q_k}$ covariance matrix and $2N_u$ degrees of freedom. In general, if the covariance matrices are linear proportional to identity matrix, then the sum of Wishart random variable follows a semi-correlated Wishart distribution \cite{Ivrlac2003}. However if the covariance matrices are not proportional to the identity matrix, then deriving the sum of Wishart distribution is nontrivial \cite{Kumar2014}. The second term also is a Wishart distribution as follow, 
\begin{align}
{{\bf{G}^u}^{\mathcal{H}}{\bf{w}^u}^\mathcal{H}\bf{w}^u}\bf{G}^u=\bf W^p\sim  \bf \mathcal {W}_{N_u}(N_tK,\bf I_{qk}).
\end{align}
Therefore, the sum of two wighted central Wishart random variable where approximation for the distribution can be obtained as follows,  
\begin{align}
&\sum_{i=1, i\neq k}^{K}{\bf{F}_k^{\mathcal{H}}\bf {V_k^d}^{\mathcal{H}}{\bf{w}_i^d}^{\mathcal{H}}\bf{w}_i^d}\bf {V_k^d}\bf{F}_k+{{\frac{p^u}{p^d}}{\bf{G}^u}^{\mathcal{H}}{\bf{w}^u}^\mathcal{H}\bf{w}^u}\bf{G}^u \notag \\ & \sum _{i=1,i\neq k}^K \bf W^s_i+ \frac{p^u}{p^d} \bf{W}^p= \boldsymbol\eta^s \bf W^q,
\end{align}
where, $\boldsymbol\eta^s=\frac{K(1+(\frac{p^u}{p^d})^2)-1}{K(1+\frac{p^u}{p^d})-1}$ and $\bf W^q \sim \bf \mathcal W_{N_u}(N_s,\bf I_{qk})$ with the following degree of freedom,
\begin{align}
N_s=\frac{N_t(\frac{p^u}{p^d}K+K-1)^2}{(\frac{p^u}{p^d})^2K+K-1},
\end{align}
therefore $\boldsymbol\Psi$ is given by,
\begin{align}
&\boldsymbol \Psi = \boldsymbol\eta^s \bf W^q+ \boldsymbol \sigma_k^{d0}= \boldsymbol \eta_k^v W_k^v,
\end{align}
where, $\boldsymbol\eta^v=\frac{\boldsymbol\eta^s N_s}{N_s+\sigma_k^{d0}}$ and $\bf W_k^v \sim \bf \mathcal W_{N_u}(2N_v,\bf I_{qk})$ with the following degree of freedom,
\begin{align}
N_v=N_s/2+\frac{\sigma_k^{d0}(2N_s+\sigma_k^{d0})}{2N_s}.
\end{align}
According to the definition of $\boldsymbol \Phi$ and $\boldsymbol \Psi$, Olkin in \cite{Ivrlac2003} show that if the $\boldsymbol \Psi^{1/2}$ is selected as a symmetric square root of  $\boldsymbol \Psi$ and also if the covariance matrices of $\boldsymbol \Phi$ and $\boldsymbol \Psi$, are proportional to identity matrix then we have the relationship among them as 
\begin{align}
\boldsymbol\gamma_k^d \!\!=\!\! \frac{\boldsymbol \Phi}{\boldsymbol \Psi}\!\!=\!\!\boldsymbol \Psi^{-1/2}\boldsymbol \Phi \boldsymbol \Psi^{-1/2}\!\!=\!\!\frac{1}{\boldsymbol \eta^v}\bf {W^v}^{-1/2}{W^d}{W^v}^{-1/2}.
\end{align}
\end{proof}
From \cite{citeulike:624837}, $\bf {W^v}^{-1/2}{W^d}{W^v}^{-1/2}\sim \bf B_{N_u}^{II}(N_t,N_v)$ distributed as a multivariate Betta type II distribution. By invoking the first and second moment, the distribution of $\boldsymbol\gamma_k^d\sim \bf B_{N_t}^{II}(N1,N2)$, where $N1$ and $N2$ are the degrees of freedom given as follows,
\begin{align}
&N1=\frac{N_t(N_t+(N_v-2)\eta^v+1)}{\eta^v(N_t+N_v-1)},\\ & N2=\frac{N_v(N_t-3\eta^v+2)+N_v^2\eta^v+2(\eta^v-1)}{N_t+N_v-1},
\end{align}
and the pdf of the distribution is given in eq. \ref{gamma_d_p}.
\subsection{Downlink with imperfect CSI }
Similar to the uplink case, we consider imperfect CSI for CSI. The actual
and the estimated channel are defined as, 
\begin{align}
\bf F_k=\sqrt{1-\alpha^2}\bf \hat{F}_k+\alpha \bf \Delta_k^d,
\end{align}
where similarly, $\hat{F}_k\sim \mathcal{C}\mathcal{N}(0,\bf I)$ is the downlink imperfect estimated channel with zero mean and unit variance and $\Delta_k^d \sim \mathcal{C}\mathcal{N}(0, \bf I)$ is the estimated downlink channel Gaussian noise. The received signal at the $k_{th}$ SU will be as follows,
\begin{align} 
\bf{y}_{k}^{d}\triangleq&~\underbrace {\sqrt{p^d}{(\sqrt{1-\alpha^2}\bf \hat{F}_k+\alpha \bf \Delta_k^d)\bf {V_k^d}\bf{w}_k^d\bf{s}_k^d}} \notag \\[2pt]&+\,\underbrace {\displaystyle \sum _{i\neq k, i \in K} \sqrt{p^d}{(\sqrt{1-\alpha^2}\bf \hat{F}_i+\alpha \bf \Delta_i^d)\bf {V_i^d}\bf{w}_i^d\bf{s}_i^d}}_{{\sf\text { DL interference}}} \notag \\[2pt]&+\,\underbrace {\displaystyle \sqrt{p^u}{\bf{G}^u\bf{w}^u\bf{s}^u}}_{{\sf\text { UL intracell interference}}} + \bf{n}_k^d. 
\end{align}
Similar to the perfect CSI case the received ID and EH signal elements are respectively denoted as,
\begin{align}
\textbf{y}_k^{di}=&\sqrt{\rho_k} [ {\sqrt{p^d}{(\sqrt{1-\alpha^2}\bf \hat{F}_k+\alpha \bf \Delta_k^d)\bf {V_k^d}\bf{w}_k^d\bf{s}_k^d}} +\notag \\& {\displaystyle \sum _{i\neq k, i \in K} \sqrt{p^d}{(\sqrt{1-\alpha^2}\bf \hat{F}_i+\alpha \bf \Delta_i^d)\bf {V_i^d}\bf{w}_i^d\bf{s}_i^d}} \notag \\&+\, {\displaystyle \sqrt{p^u}{\bf{G}^u\bf{w}^u\bf{s}^u}}]+\bf{n}_k^s,
\label{eq:d_i}
\\
\textbf{y}_k^{de}=&\sqrt{1-\rho_k}[ {\sqrt{p^d}{(\sqrt{1-\alpha^2}\bf \hat{F}_k+\alpha \bf \Delta_k^d)\bf {V_k^d}\bf{w}_k^d\bf{s}_k^d}} +\notag \\ & {\displaystyle \sum _{i\neq k, i \in K} \sqrt{p^d}{(\sqrt{1-\alpha^2}\bf \hat{F}_i+\alpha \bf \Delta_i^d)\bf {V_i^d}\bf{w}_i^d\bf{s}_i^d}} \notag \\&+\, {\displaystyle \sqrt{p^u}{\bf{G}^u\bf{w}^u\bf{s}^u}}+\textbf{n}_k^d].
\label{eq:d_e}
\end{align}
\begin{prep}
For multi user downlink FD MIMO system with imperfect CSI and the processed SINR follows a matrix variate Beta type $II$ distribution with parameters $(a,b)$ and defined as,
\begin{align}
&\boldsymbol\gamma_k^d\approx B_{N_u}^{II}(N1,N2)\sim \notag \\ &\frac{\det(\boldsymbol \gamma_k^d)^{(\frac{2N1-N_u-1}{2})}\det{(\bf I_{\hat q_k}+\boldsymbol{\gamma}_k^d)^{{-N1-N2}}}}{\beta( N1,N2)},
\label{gamma_d}
\end{align}
where \vspace{-1em}
\begin{align}
\beta(N1,N2)= \frac{\Gamma_{N_u}(N1)\Gamma_{N_u}(N2)}{\Gamma_{N_u}({N1+N2})},
\end{align}
and $\Gamma_{N_u}(x)$ is the multivariate Gamma function given as,
\begin{align}
\Gamma_{N_u}(x)=\pi^{{N_u}({N_u}-1)/4}\prod_{i=1}^{N_u}\Gamma(x-(i-1)/2).
\end{align}
\end{prep}
\begin{figure}[t!]
\centering
    \includegraphics[width=0.52\textwidth]{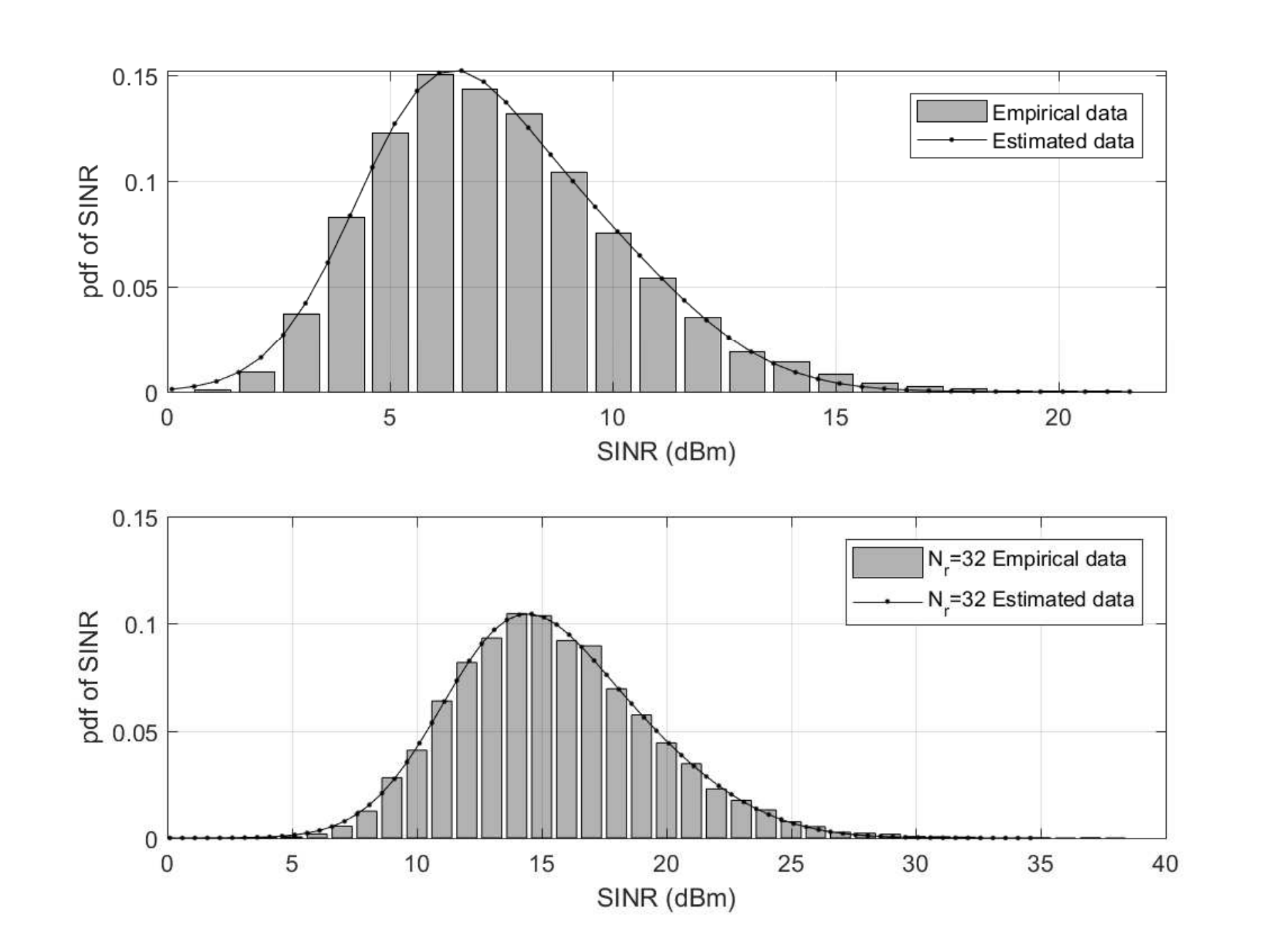}
  \caption{PDF of perfect CSI uplink SINR.}
  \label{fig:4}
\centering
    \includegraphics[width=0.52\textwidth]{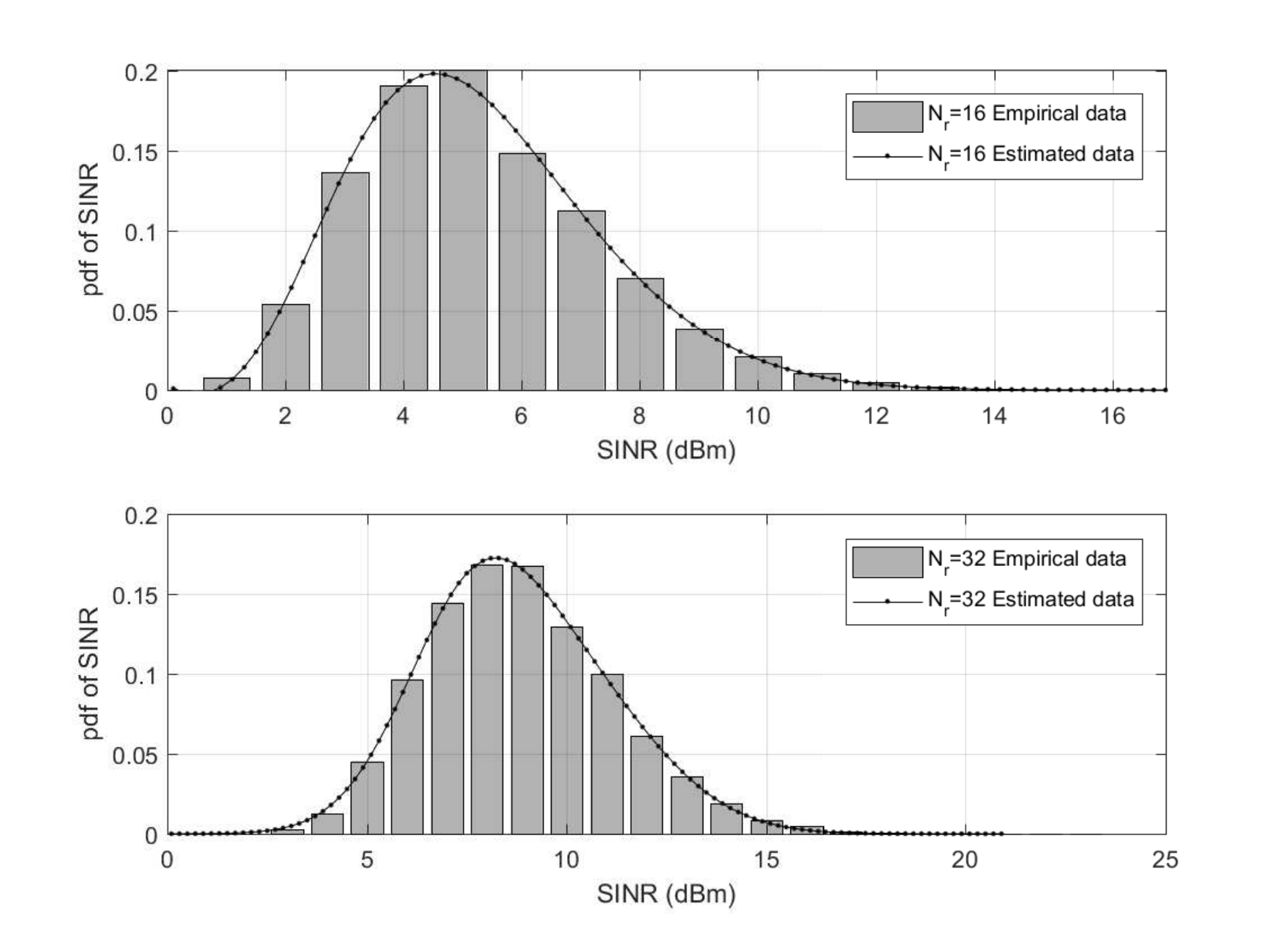}
  \caption{PDF of imperfect CSI uplink SINR.}
  \label{fig:5}
 \end{figure}
\begin{figure}[t!]
\centering
    \includegraphics[width=0.52\textwidth]{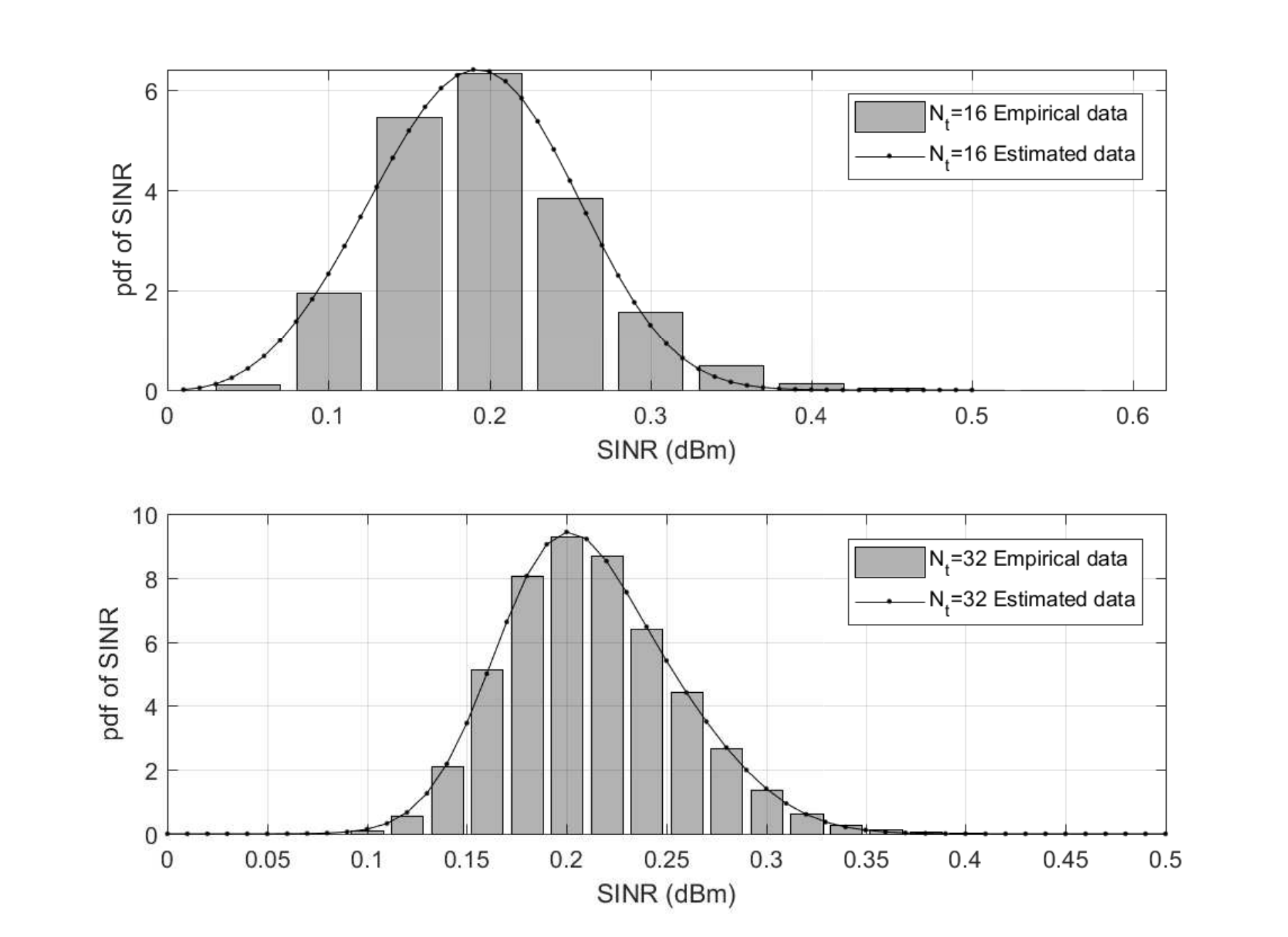}
  \caption{PDF of perfect CSI downlink SINR.}
  \label{fig:4}
\end{figure}
\begin{figure}[t!]
\centering
    \includegraphics[width=0.52\textwidth]{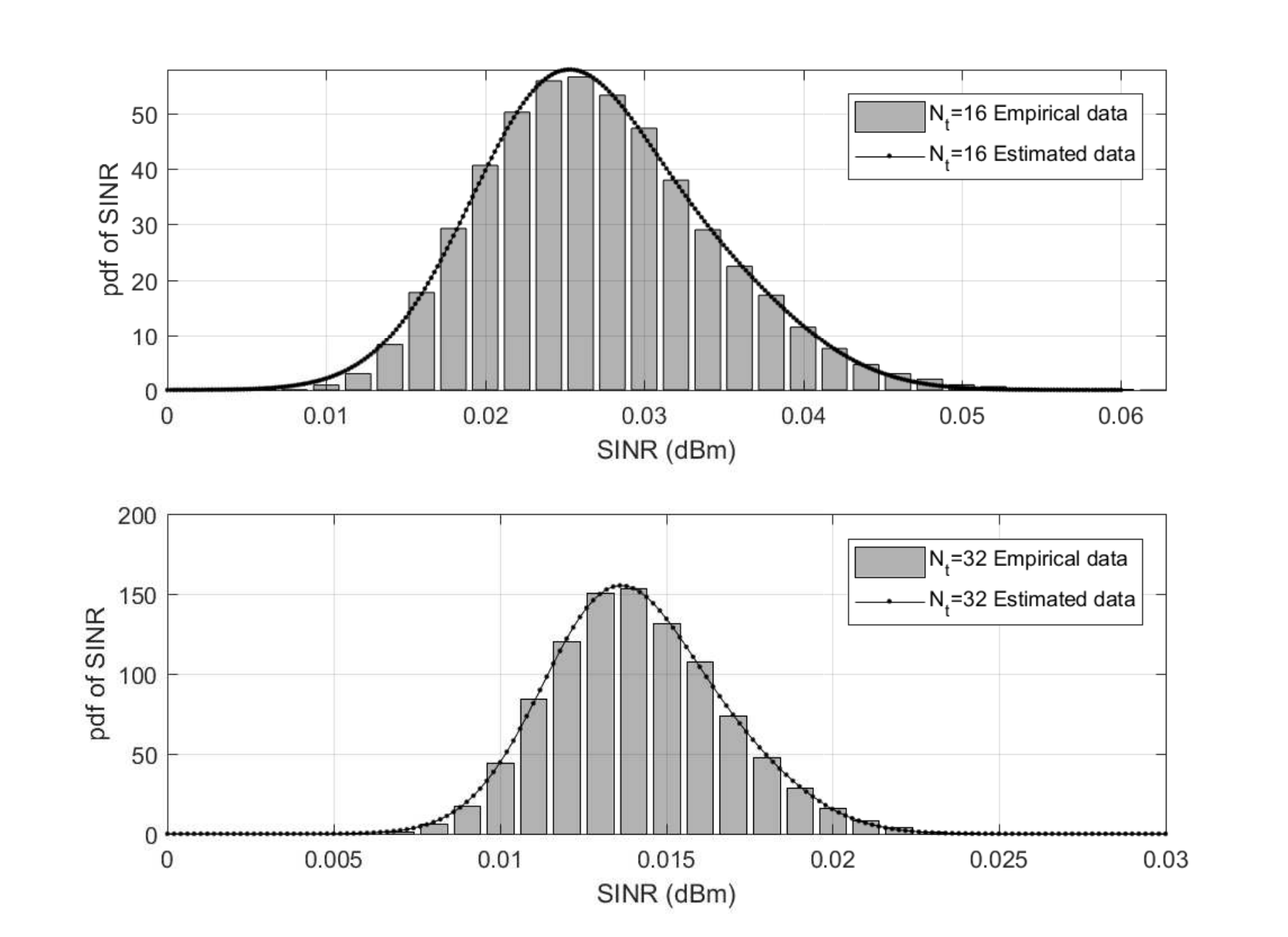}
  \caption{PDF of imperfect CSI downlink SINR.}
  \label{fig:5}
 \end{figure}
\begin{proof}
Similar to the perfect case, for simplicity we denote the received SINR of MU-MIMO with energy harvesting for $k_{th}$ user at  eq.\ref{eq:d_sinr4} as $\boldsymbol\gamma_k^d=\frac{\boldsymbol\hat \Phi}{\boldsymbol\hat\Psi}$ and ${\boldsymbol\hat\sigma_k^{d_0}}=\frac{{\boldsymbol\sigma_k^d}^2}{(1-\alpha^2){p^d}}+\frac{{\boldsymbol\sigma_k^s}^2}{\rho_k(1-\alpha^2){p^d}}$. Therefore, $\hat \Phi_k$ has a central multivariate matrix Wishart distribution with $\bf I_{\hat q_k}$ covariance matrix and $2N_t$ degrees of freedom as follows,
 \begin{align}
&\boldsymbol\hat\Phi_k=\bf \hat W_k^d\sim \bf \mathcal W_{N_u}(2N_t,\bf I_{\hat q_k})=\notag\\&\frac{{\gamma_{k}^{d}}^{(2N_t-N_u-1)/2}exp\left[-\frac{1}{2}tr(\bf ( I_{\hat q_k})^{-1}\gamma_{k}^d)\right]}{2^{N_u N_t}\Gamma_{N_{u}}(N_t)\det( \bf I_{\hat q_k})^{N_t}},
\end{align}
where $\bf \mathcal W_x(y,z)$ is a central multivariate Wishart distribution. Furthermore, the the first term of $\hat\Psi$ in eq.\ref{eq:d_sinr4}, is a sum of central Wishart distribution with similar identity covariance matrices can be simplified as follow,\vspace{-0.5em}
\begin{align}
&\sum_{i=1, i\neq k}^{K}{\bf{\hat F}_k^{\mathcal{H}}\bf {V_k^d}^{\mathcal{H}}{\bf{w}_i^d}^{\mathcal{H}}\bf{w}_i^d\bf{V}_i^d}\bf{\hat F}_k=\bf \hat W^s,
\end{align}
where $\bf \hat W^s \sim \bf \mathcal W_{N_u}(2N_t(K-1),\bf I_{\hat q_k})$ has a multivariate Wishart distribution.
Similarly, in the second term of $\hat \Psi$ in eq.\ref{eq:d_sinr4}, is a sum of central Wishart distribution with similar identity covariance matrices can be simplified as follow,\vspace{-0.5em}
\begin{align}
&\sum_{i=1}^{K}{\bf{{\Delta}_i^d}^{\mathcal{H}}\bf {V_i^d}^{\mathcal{H}}{\bf{w}_i^d}^{\mathcal{H}}\bf{w}_i^d}\bf {V_i^d}\bf{{\Delta}_i^d}= \bf \hat W^{\Delta}
\end{align}
where $\bf\hat W^{\Delta} \sim \bf \mathcal W_{N_u}(2N_tK,\bf I_{\hat q_k})$ has a multivariate Wishart distribution with $\bf I_{\hat q_k}$ covariance matrix and $2N_t$ degrees of freedom.
The third term also is a Wishart distribution as follow, \vspace{-0.5em}
\begin{align}
{{\bf{G}^u}^{\mathcal{H}}{\bf{w}^u}^\mathcal{H}\bf{w}^u}\bf{G}^u=\bf \hat W^p\sim  \bf \mathcal {W}_{N_u}(2N_tK,\bf I_{\hat q_k}).
\end{align}
Therefore, the sum of three wighted central Wishart random variable where approximation for the distribution can be obtained as follows,  
\begin{align}
&\sum_{i=1, i\neq k}^{K}{\bf{\hat F}_k^{\mathcal{H}}\bf {V_k^d}^{\mathcal{H}}{\bf{w}_i^d}^{\mathcal{H}}\bf{w}_i^d}\bf {V_k^d}\bf{\hat F}_k+{{\frac{p^u}{p^d}}{\bf{G}^u}^{\mathcal{H}}{\bf{w}^u}^\mathcal{H}\bf{w}^u}\bf{G}^u \notag \\ & \sum _{i=1,i\neq k}^K \bf\hat W^s_i+ \frac{p^u}{(1-\alpha^2)p^d} \bf\hat {W}^p= \boldsymbol\hat\eta^g \bf \hat W^g,
\end{align}
where, $\boldsymbol\hat\eta^g=\frac{K(1+(\frac{p^u}{(1-\alpha^2)p^d})^2)-1}{K(1+\frac{p^u}{p^d})-1}$ and $\bf \hat W^g \sim \bf \mathcal W_{N_u}(N_g,\bf I_{\hat q_k})$ with the following degree of freedom,
\begin{align}
N_g=\frac{2N_t(\frac{p^u}{(1-\alpha^2)p^d}K+K-1)^2}{(\frac{p^u}{(1-\alpha^2)p^d})^2K+K-1},
\end{align}
and  \vspace{-1em}
\begin{align}
\boldsymbol\hat \eta^g \bf \hat W^g+\frac{\alpha^2}{1-\alpha^2}\bf \hat W^{\boldsymbol \Delta}= \boldsymbol\hat\eta^q \bf \hat W^q,
\end{align}
where, $\boldsymbol\hat\eta^q=\frac{N_g\hat{\boldsymbol{\eta}^g}^2+2N_tK(\frac{\alpha^2}{1-\alpha^2})^2}{N_g\hat{\boldsymbol\eta^r}+2N_tK\frac{\alpha^2}{1-\alpha^2}}$ and $\bf \hat W^q \sim \bf \mathcal W_{N_u}(N_q,\bf I_{\hat q_k})$ with the following degree of freedom,
\begin{align}
N_q=\frac{(\frac{2N_tK\alpha^2}{(1-\alpha^2)}+N_g\boldsymbol\hat\eta^g)^2}{2N_tK(\frac{\alpha^2}{(1-\alpha^2)})^2+N_g\hat{\boldsymbol\hat\eta^r}^2},
\end{align}
therefore $\boldsymbol\hat\Psi$ is given by,
\begin{align}
&\boldsymbol \hat\Psi = \boldsymbol\hat\eta^q \bf\hat W^q+ \boldsymbol \hat\sigma_k^{d0}= \boldsymbol \hat\eta^v \hat W^v,
\end{align}
where, $\boldsymbol\hat\eta^v=\frac{\boldsymbol\hat\eta^q N_q}{N_q+\hat\sigma_k^{d0}}$ and $\bf\hat W^v \sim \bf \mathcal W_{N_u}(2N_v,\bf I_{\hat q_k})$ with the following degree of freedom,
\begin{align}
N_v=N_q/2+\frac{\hat\sigma_k^{d0}(2N_q+\hat\sigma_k^{d0})}{2N_q}. 
\end{align}
\end{proof}
Similar to the perfect CSI case, we have,
\begin{align}
&N1=\frac{N_t(N_t+(N_v-2)\hat\eta_k^v+1)}{\hat\eta_k^v(N_t+N_v-1)},\\ & N2=\frac{N_v(N_t-3\hat\eta_k^v+2)+N_v^2\hat\eta_k^v+2(\hat\eta_k^v-1)}{N_t+N_v-1},
\end{align}
and the pdf of the distribution is given in eq. \ref{gamma_d}.
\section{SIMULATION RESULTS}
In this section, we study the performance of the proposed closed-form approximation of the SINR of the downlink SWIPT-MIMO  with antenna selection and the uplink MU-MIMO with ZFBF.  we assume that the number of receive and transmit antennas at the AGG is equal and larger than the number of SUs. The total bandwidth is 10 MHz. We also consider 3 SUs facilitated with 2 transceiver antennas. The channel uncertainty is considered as $\alpha=0.2$. The power split factor is assumed to be $\rho=0.3$ and the power efficiency at the EH unit is $\eta_k=40\%$ for all SUs and the harvested energy is stored in a 3.2 v 20Ah battery. 
Fig. \ref{fig:4} and Fig. \ref{fig:5} show the empirical data and chi-square distribution as a simplified form of the proposed Wishart distribution for a different number of receive antennas for perfect an imperfect CSI. Note that the number of received antennas increases, the distributions $pdf$ moves to the higher SINR region and the estimation becomes more accurate.  
Fig. \ref{fig:4} and Fig. \ref{fig:5} show the empirical data and Beta distribution of type II as a simplified form of the proposed multivariate Beta distribution of type II  for a different number of receive antennas of the downlink SWIPT MU-MIMO with antenna selection for perfect an imperfect CSI. Numerical results show that the proposed closed-form approximation offers a perfect match with the empirical data in a wide range of the number of antennas.
\section{Conclusion}
In this paper a new low complexity closed-form approximation of the Full-duplex SWIPT MU-MIMO system SINR distribution for the received signal at the sensor users with perfect and imperfect CSI and transmit antenna selection scheme at the transmitter is proposed. SINR distribution of the uplink with ZFBF and antenna selection at the AGG receiver with transmit antenna selection also studied. The uplink SINR with perfect and imperfect CSI is modeled with multivariate Wishart distributions and the downlink SINR with perfect and imperfect CSI is modeled by a multivariate Beta type II distribution. The proposed SINR distributions analytical results are compared to the  Monte-Carlo simulations and we obtained a perfect match.
\bibliographystyle{IEEEtran}
{\footnotesize
\bibliography{IEEEabrv,Ref_EE_FD}}

\end{document}